\documentclass[11pt,reqno]{amsart}
\newcommand\version{September 26, 2007}

\font\notefont=cmsl8  
\usepackage{amsxtra}
\usepackage{amsmath,amsfonts,amsthm,amssymb}

\newtheorem{theorem}{Theorem}[section]
\newtheorem{proposition}[theorem]{Proposition}
\newtheorem{lemma}[theorem]{Lemma}
\newtheorem{corollary}[theorem]{Corollary}

\theoremstyle{definition}

\theoremstyle{remark}

\newtheorem{remark}[theorem]{Remark}

\numberwithin{equation}{section}

\newcommand{\cl}{\mathrm{cl}}

\renewcommand{\epsilon}{\varepsilon}

\newcommand{\N}{\mathbb{N}}

\renewcommand{\phi}{\varphi}

\newcommand{\R}{\mathbb{R}}

\newcommand{\Z}{\mathbb{Z}}

\DeclareMathOperator{\curl}{curl}

\DeclareMathOperator{\tr}{tr}

\title{P\'olya's conjecture in the presence of a constant magnetic field}

\author{Rupert L. Frank}
\address{Rupert L. Frank, Department of Mathematics, Fine Hall,
  Princeton University, Princeton, NJ 08544, USA}
\email{rlfrank@math.princeton.edu}
\author{Michael Loss}
\address{Michael Loss, School of Mathematics, Georgia Institute of
  Technology, Atlanta, GA 30332-0160, USA}
\email{loss@math.gatech.edu}
\author{Timo Weidl}
\address{Timo Weidl, Department of Mathematics and Physics, Stuttgart
  University, Pfaffen\-wald\-ring 57, 70569 Stuttgart, Germany}
\email{weidl@mathematik.uni-stuttgart.de}

\begin{document}

\begin{abstract}
We consider the Dirichlet Laplacian with a constant magnetic field in
a two-dimensional domain of finite measure. We determine the sharp
constants in semi-classical eigenvalue estimates and show, in
particular, that P\'olya's conjecture is not true in the presence of a
magnetic field.
\end{abstract}

\thanks{\copyright\, 2007 by the authors. This paper may be reproduced, in its
entirety, for non-commercial purposes.}

\maketitle


\section{Introduction}

Let $\Omega\subset\R^2$ be a domain of finite measure and define the
Dirichlet Laplacian $H^\Omega$ in $L_2(\Omega)$ as the Friedrichs
extension of $-\Delta$ initially given on $C_0^\infty(\Omega)$. This
defines a self-adjoint non-negative operator, and by Rellich's
compactness theorem its spectrum is discrete and accumulates at
infinity only.
The spectrum of $H^\Omega$ plays an important role
in many physical models (such as membrane vibration or quantum
mechanics) and its determination is a classical problem in
mathematical physics.

Let $(\lambda_n)$ be the non-decreasing sequence of eigenvalues of
$H^\Omega=-\Delta$ (taking multiplicities into account) and let
$N(\lambda,H^\Omega):=\#\{n:\ \lambda_n<\lambda\}$ denote their counting
function. 

In 1911 H. Weyl \cite{W} (see also \cite[Ch. XIII]{RS}) showed the
asymptotic formula
\begin{equation*} 
\lambda_n=\frac{4\pi n}{|\Omega|}(1+o(1)),\quad n\to\infty.
\end{equation*}
In terms of the counting function this is equivalent to
\begin{equation}\label{eq:weyl}
  N(\lambda,H^\Omega) 
  = \frac1{4\pi}\lambda |\Omega|(1+o(1)),\quad \lambda\to+\infty\,.
\end{equation}
Integrating the latter formula one finds as well the asymptotic
behavior of the eigenvalue means
\begin{equation}\label{eq:bly}
	\tr(H^\Omega-\lambda)_-^\gamma :=\sum_{n:\lambda>\lambda_n}
        (\lambda-\lambda_n)^\gamma
        = L_{\gamma,2}^{\cl} \lambda^{\gamma+1} |\Omega| (1+o(1)),
	\quad \lambda\to+\infty,
\end{equation}
where $\gamma\geq 0$ and
\begin{equation}\label{eq:classconst2}
	L_{\gamma,2}^{\cl} := \left(4\pi(\gamma+1)\right)^{-1}.
\end{equation}
Note that the expression on the right hand side of \eqref{eq:bly}
equals the classical phase space average
\begin{equation}\label{eq:phspv}
 L_{\gamma,2}^{\cl} \lambda^{\gamma+1} |\Omega|=
(2\pi)^{-2} \iint_{\Omega\times\R^2} (|\xi|^2-\lambda)_-^\gamma \,dxd\xi
\end{equation}
of the symbol $|\xi|^2$  of the Laplacian.

P\'olya \cite{P} found in 1961  that for 
{\em tiling domains}\footnote{A domain $\Omega\subset \R^2$ is tiling if
one can cover $\R^2$ up to a set of measure zero by pairwise disjoint
congruent copies of $\Omega$.} $\Omega$ the asymptotic expression
\eqref{eq:weyl} is in fact an upper bound on the counting function, namely
\begin{equation}\label{eq:polya}
	N(\lambda,H^\Omega) \leq \frac1{4\pi}\lambda |\Omega|,\quad
        \lambda\geq 0.
\end{equation}
By \eqref{eq:weyl} the constant in this bound is optimal. Moreover,
P\'olya conjectured that this bound should hold true for {\em
  arbitrary} domains $\Omega$ with the same sharp constant
$\frac1{4\pi}$.

The fact that the counting function $N(\lambda,H^\Omega)$ can be estimated by
\begin{equation}\label{eq:liebest}
	N(\lambda,H^\Omega) \leq C\lambda |\Omega|,\quad
        \lambda\geq 0.
\end{equation}
with some constant $C$ which does not depend on $\lambda$ or the shape
of the domain is due to Rozenblum \cite{R1}, Lieb \cite{Li2} and
Metivier \cite{M}. Results with sharp constants for \emph{sums} of
eigenvalues have been obtained by Berezin and by Li and Yau. Indeed,
Berezin \cite{B} proved that
\begin{equation}\label{eq:br}
	\tr(H^\Omega-\lambda)_-^\gamma \leq L_{\gamma,2}^{\cl} \lambda^{\gamma+1} |\Omega|
	\quad\mbox{for}\quad \gamma\geq 1\,.
\end{equation}
In view of the Weyl asymptotics \eqref{eq:bly} the constant in this
bound is optimal. This estimate in the case $\gamma=1$ implies after
taking the Legendre transform the celebrated result by Li and Yau \cite{LY}
\begin{equation}\label{eq:ly}
\sum_{j=1}^n \lambda_j \geq \frac{2\pi n^2}{|\Omega|},
\quad n\in \mathbb N\,.  
\end{equation}
Both \eqref{eq:br} and \eqref{eq:ly} give rise to the best known upper
bound $C\leq (2\pi)^{-1}$ on the sharp constant $C$ in
\eqref{eq:liebest}. However, P\'olya's conjecture, namely that
\eqref{eq:polya} holds for general domains, remains open. In fact,
this question is unresolved even in the case where the domain is a disk.

The main goal of this paper is to {\em disprove} the analogous conjecture for
the Dirichlet Laplacian with a constant magnetic field.

Put $D=-i\nabla$ and let $A$ be a sufficiently regular real vector field on
$\Omega$. We consider the operator $(D-A)^2$ on $L_2(\Omega)$ with
Dirichlet boundary conditions defined in the quadratic form sense. If
$|\Omega|$ has finite measure, the spectrum of $(D-A)^2$ is discrete
and as above, we can introduce the ordered sequence of eigenvalues and
the corresponding counting function. It is well-known that the
asymptotic formulae \eqref{eq:weyl} and \eqref{eq:bly} remain true in
the magnetic case as well. This is in accordance with the fact that
the magnetic field leaves the classical phase space average unchanged,
\begin{equation*}
	(2\pi)^{-2} \iint_{\Omega\times\R^2} (|\xi-A(x)|^2-\lambda)_-^\gamma \,dxd\xi
	= (2\pi)^{-2} \iint_{\Omega\times\R^2} (|\xi|^2-\lambda)_-^\gamma \,dxd\xi\,.
\end{equation*}
Therefore, it seems reasonable to discuss P\'olya-type bounds in the magnetic
case as well. In fact, it turns out that the bound \eqref{eq:liebest} extends
to the magnetic case with a suitable constant $C$ which does not depend
on $A$, $\Omega$ and $\lambda$, see e.g. \cite{R2}. 

There are also results concering magnetic estimates with \emph{sharp}
semi-classical constants. As recalled in the appendix, a result by
Laptev and Weidl \cite{LW} implies the bound
\begin{equation}\label{eq:blymag}
	\tr((D-A)^2-\lambda)_-^\gamma \leq L_{\gamma,2}^{\cl} \lambda^{\gamma+1} |\Omega|
\end{equation}
for \emph{arbitrary} $A$ and all $\gamma\geq 3/2$. 
In \cite{ELV} this result was extended to $\gamma\geq 1$ in the special case
of a homogeneous magnetic field, $A(x) = \frac B2 (-x_2,x_1)^T$.  The
latter two results motivate the question, whether P\'olya's conjecture
could be true in the magnetic case.

In this note we shall show that this intuition is wrong and that the
P\'olya estimate \eqref{eq:polya} in the magnetic case can be violated
even for tiling domains. More precisely, we consider a homogeneous
magnetic field, $A(x) = \frac B2 (-x_2,x_1)^T$, and show that
for arbitrary domains $\Omega$ of finite measure the bound
\begin{equation}\label{eq:main}
	N(\lambda,(D-A)^2) \leq \frac1{2\pi}\lambda |\Omega| = 2 L_{0,2}^{\cl} \lambda |\Omega|
\end{equation}
holds true. We prove that the constant in this bound is optimal and that
the numerical factor~$2$ on the right hand side cannot be improved -
not even in the tiling case. A similar phenomenon occurs for
eigenvalue moments of order $\gamma\in(0,1)$.

As a consequence of our result we see, in particular, that any attempt
to prove P\'olya's conjecture with a method which extends to constant
magnetic fields must fail.



\section{Main results}\label{sec:main}

Let $\Omega\subset\R^2$ be a domain of finite measure. For $B>0$ we consider the self-adjoint operator 
\begin{equation*}
	H_B^\Omega := (D-BA)^2
	\qquad\text{in } L_2(\Omega)
\end{equation*}
with Dirichlet boundary conditions, i.e., closing the form $\|(D-BA)u\|^2$ on $C_0^\infty(\Omega)$. The magnetic vector potential $A$ is always chosen in the form
\begin{equation*}
	A(x) := \frac12 (-x_2,x_1)^T,
\end{equation*}
and we remark that $\curl BA\equiv B$.
In other words, we restrict the vector potential for 
a constant magnetic field from $\R^2$ to $\Omega$.
\footnote{
For simply connected domains 
$\Omega$ this choice of $A$ is up to gauge invariance
unique in the class of all vector potentials inducing a constant magnetic
field in $\Omega$. If $\Omega$ is not simply connected, then one has 
gauge invariant classes
of magnetic vector potentials inducing a constant magnetic field {\rm inside}
$\Omega$, but which are not restrictions of a vector potential producing
a constant magnetic field on the whole of $\R^2$. In this paper we do not
consider such vector potentials.}

The operator $H_B^\Omega$ has compact resolvent and we denote by $N(\lambda, H_B^\Omega)$ the number of its eigenvalues less than $\lambda$, counting multiplicities. Our first main result is

\begin{theorem}\label{main1}
	Let $\Omega\subset\R^2$ be a domain of finite measure. Then for all $B>0$ and $\lambda>0$
	\begin{equation}\label{eq:main1number}
		N(\lambda, H_B^\Omega) \leq R_0 L_{0,2}^{\cl} |\Omega| \lambda
	\end{equation}
	and
	\begin{equation}\label{eq:main1moments}
		\tr(H_B^\Omega-\lambda)_-^\gamma \leq R_\gamma L_{\gamma,2}^{\cl} |\Omega| \lambda^{\gamma+1},
		\qquad 0<\gamma<1,
	\end{equation}
	where $R_0=2$ and $R_\gamma=2\left(\gamma/(\gamma+1)\right)^\gamma$ for $0<\gamma<1$. One has $R_\gamma>1$ and these constants can not be improved, not even if $\Omega$ is tiling. More precisely, for any $0\leq\gamma<1$, $\epsilon>0$ and $B>0$ there exists a square $\Omega$ and $\lambda>0$ such that
	\begin{equation}\label{eq:main1counter}
		\tr(H_B^\Omega-\lambda)_-^\gamma \geq (1-\epsilon) R_\gamma L_{\gamma,2}^{\cl} |\Omega| \lambda^{\gamma+1}.
	\end{equation}
\end{theorem}

We emphasize that for linear and superlinear moments one has the semi-classical bound
\begin{equation}\label{eq:elv}
	\tr(H_B^\Omega-\lambda)_-^\gamma \leq L_{\gamma,2}^{\cl} |\Omega| \lambda^{\gamma+1},
	\qquad \gamma\geq 1,
\end{equation}
\emph{without} an excess factor. The inequality \eqref{eq:elv} is essentially contained in \cite{ELV} but will be rederived in Corollary \ref{elv} below.

Our second main result concerns tiling domains. We shall show that in this case the inequalities \eqref{eq:main1number} and \eqref{eq:main1moments} can be strengthened if one is willing to allow the right hand side depend on $B$. Let us define
\begin{equation}\label{eq:magweyl}
\mathfrak B_\gamma(B,\lambda) 
:= (2\pi)^{-1} B \sum_{k\in\N_0}
\left(\lambda-B(2k+1)\right)_+^{\gamma}.
\end{equation}
For $\gamma=0$ this is defined to be left-continuous in $\lambda$, i.e., $0^0_-:=0$.

\begin{theorem}\label{main2}
	Let $\Omega\subset\R^2$ be a tiling domain of finite measure. Then for all $B>0$ and $\lambda>0$
	\begin{equation}\label{eq:main2number}
		N(\lambda, H_B^\Omega) \leq \mathfrak B_0(B,\lambda) |\Omega|
	\end{equation}
	and
	\begin{equation}\label{eq:main2moments}
		\tr(H_B^\Omega-\lambda)_-^\gamma \leq \mathfrak B_\gamma(B,\lambda) |\Omega|,
		\qquad 0<\gamma<1,
	\end{equation}
	and these estimates can not be improved. More precisely, for any $0\leq\gamma<1$, $\epsilon>0$, $B>0$, $\lambda>0$ there exists a square $\Omega$ such that
	\begin{equation}\label{eq:main2counter}
		\tr(H_B^\Omega-\lambda)_-^\gamma \geq (1-\epsilon) \mathfrak B_\gamma(B,\lambda) |\Omega|.
	\end{equation}
\end{theorem}

We emphasize that for $\gamma\geq 1$ one has the bound
\begin{equation}\label{eq:elv2}
	\tr(H_B^\Omega-\lambda)_-^\gamma \leq \mathfrak B_\gamma(B,\lambda) |\Omega|,
	\qquad \gamma\geq 1,
\end{equation}
for an \emph{arbitrary} domain $\Omega\subset\R^2$ of finite measure. This is again essentially contained in \cite{ELV}. We give an independent proof in Theorem \ref{elv2} below and show also that \eqref{eq:elv2} is stronger than \eqref{eq:elv}. The question whether \eqref{eq:main2number} and \eqref{eq:main2moments} extend to not necessarily tiling domains is left open.

\begin{remark}
	There are estimates intermediate between \eqref{eq:main1number} and \eqref{eq:main2number} with the right hand side depending on $B$ but in a simpler way than in \eqref{eq:main2number}. For example, we mention the estimate
	\begin{equation}\label{eq:main1number0}
		N(\lambda,H_B^\Omega) \leq \frac1{4\pi}(\lambda+B) |\Omega|
	\end{equation}
	for $\Omega$ tiling. Note that this estimate is \emph{stronger} than \eqref{eq:main1number} since $N(\lambda,H_B^\Omega)=0$ for $\lambda\leq B$. In particular, it coincides with the estimate \eqref{eq:polya} for $B=0$.
\end{remark}

\begin{remark}
	There is an essentially equivalent way of stating the estimates \eqref{eq:main1number} and \eqref{eq:main1number0}. Namely denoting the eigenvalues of $H_B^\Omega$ by $\lambda_{B,j}^\Omega$ and passing to the limit $\lambda\to \lambda_{B,j}^\Omega+$ in these estimates we find
	\begin{equation*}
		\lambda_{B,N}^\Omega \geq 2\pi |\Omega|^{-1} N
	\end{equation*}
	and, respectively,
	\begin{equation*}
		\lambda_{B,N}^\Omega \geq 4\pi |\Omega|^{-1} N - B.
	\end{equation*}
\end{remark}

\begin{remark}
	For the lower bound \eqref{eq:main1counter} we fix $B>0$ and choose $\Omega$ and $\lambda$. Alternatively, one can fix a cube $\Omega$ and choose $B$ and $\lambda$. This follows by a simple scaling argument.
\end{remark}
	

\section{The magnetic density of states}

\subsection{The magnetic density of states}

In this section we shall use a slightly modified notation. When $\Omega=(-L/2,L/2)^2$ we shall denote the operator $H_B^\Omega$ by $H_B^D(L)$. Recall that $\mathfrak B_0(B,\lambda)$ was defined in \eqref{eq:magweyl}. Our goal is to prove

\begin{proposition}\label{dos}
Let $B>0$ and $\lambda>0$. Then
\begin{equation}\label{eq:dos}
	\lim_{L\to\infty} L^{-2} N(\lambda,H_B^D(L)) = \mathfrak B_0(B,\lambda).
\end{equation}
\end{proposition}

Hence $\mathfrak B_0(B,\cdot)$ is the density of states for the Landau
Hamiltonian $H_B := (D-BA)^2$ in $L_2(\R^2)$. This is certainly
well-known, but we include the proof for the sake of
completeness. This will be done in the remaining part of this section.
A different proof may be found in \cite{N}. Alternatively, one can also use the known result that
\begin{equation*}
	\lim_{L\to\infty} L^{-2} N(\lambda,H_B^D(L)) = \lim_{L\to\infty} L^{-2} \tr (\chi_{Q_L} \chi_{(0,\lambda)}(H_B)).
\end{equation*}
The RHS can be evaluated using the explicit form of the spectral projections of $H_B$, see the proof of Theorem \ref{elv2}.


\subsection{Explicit solution on the torus}

In this subsection we consider the case of a square,
$\Omega=(-L/2,L/2)^2=:Q_L$, and define an operator $H_B^P(L)$ in 
$L_2(Q_L)$ which differs from $H_B^\Omega$ by the choice of magnetic
periodic boundary conditions. However, its spectrum will turn out to
be explicitly computable.

To define $H_B^P(L)$ we shall fix $B,L>0$ such that
\begin{equation}\label{eq:flux}
(2\pi)^{-1} L^2 B \in\N
\end{equation}
and introduce the `magnetic translations'
\begin{align*}
	(T_1u)(x) & := e^{-iBLx_2/2} u(x_1+L,x_2),\\
	(T_2u)(x) & := e^{iBLx_1/2} u(x_1,x_2+L).
\end{align*}
(The dependence on $B$ and $L$ is not reflected in the notation.) The assumption \eqref{eq:flux} implies that $T_1$ and $T_2$ commute, and hence any function $u$ on $Q_L$ has a unique extension to a function $\tilde u$ on $\R^2$ by means of the operators $T_1$, $T_2$. We introduce the Sobolev spaces
\begin{equation*}
	H^k_{per}(Q_L) := \{ u\in H^k(Q_L) :\ \tilde u \in H^k_{loc}(\R^2) \}.
\end{equation*}
Then the operator $H_B^P(L):=(D-BA)^2$ in $L_2(Q_L)$ with domain $H^2_{per}(Q_L)$ is self-adjoint. It is generated by the quadratic form $\|(D-BA)u\|^2$ with form domain $H^1_{per}(Q_L)$. The spectrum of this operator is described in

\begin{proposition}\label{torus}
	Assume \eqref{eq:flux}. Then the spectrum of $H_B^P(L)$ consists of the eigenvalues $B(2k+1)$, $k\in\N_0$, with common multiplicity $(2\pi)^{-1} L^2 B$. In particular, for all $\lambda>0$,
	\begin{equation}\label{eq:torus}
		N(\lambda, H_B^P(L)) = L^2 \mathfrak B_0(B,\lambda).
	\end{equation}
\end{proposition}

We recall the proof from \cite{CV}.

\begin{proof}
	Consider the closed operator $Q:=(D_1-BA_1)+i(D_2-BA_2)$ with domain $H^1_{per}(Q_L)$. Its adjoint is given by $Q^*:=(D_1-BA_1)-i(D_2-BA_2)$ with domain $H^1_{per}(Q_L)$ and one has
	\begin{equation*}
		\|(D-BA)u\|^2 = \|Qu\|^2 + B \|u\|^2 = \|Q^*u\|^2 - B \|u\|^2,
		\qquad u\in H^1_{per}(Q_L).
	\end{equation*}
	Hence $H_B^P(L)=Q^*Q+B$ and $Q Q^* - Q^* Q = 2B$. By standard
        arguments using these commutation relations one computes
        the spectrum of $H_B^P(L)$ to consist of the eigenvalues
        $B(2k+1)$, 
        $k\in\N_0$, with a common multiplicity, say $m$. 
To determine $m$ we note that 
$$N(\lambda, H_B^P(L))=m\#\{k\in\N_0 :\ B(2k+1)<\lambda\}\sim
m\lambda/2B\quad\mbox{as}\quad \lambda\to\infty\,.$$ 
On the other hand, the Weyl-type asymptotics on the counting function 
holds true for the Dirichlet and the Neumann boundary conditions, 
and hence also for the periodic operator,
$$N(\lambda, H_B^P(L)) \sim \lambda L^2/4\pi\,\quad\mbox{as}\quad
\lambda\to\infty\,.$$
Comparing the two asymptotics above one finds that $m=L^2 B/2\pi$.
\footnote{Alternatively, 
we may determine $m$ using the Aharonov-Casher theorem. 
Indeed, the multiplicity 
$m$ is the dimension of the kernel of the Pauli operator 
$(\sigma\cdot(D-B A))^2$ acting on the sections of a complex line
bundle over the torus $(\R/L\Z)^2$.}
\end{proof}


\subsection{Boundary conditions}

In this subsection we shall quantify the intuition that a change of
the boundary conditions of a differential operator has only a
relatively small effect on the overall eigenvalue distribution. We
shall denote by $H_B^N(L)$ the operator $(D-BA)^2$ with (magnetic)
Neumann boundary conditions in $\Omega=Q_L=(-L/2,L/2)^2$,  that is the
operator generated by the quadratic form $\|(D-BA)u\|^2$ with form
domain $H^1(Q_L)$. We denote by $\|K\|_1 = \tr(K^*K)^{1/2}$ the trace
norm of a trace class operator $K$.

A special case of a result by Nakamura \cite{N} (who also allows for a variable magnetic field and an electric potential) is

\begin{proposition}\label{nakamura}
Let $m\in\N$ and $B>0$. Then there exists a constant $C_m(B)>0$ such that for all $L\geq 1$
\begin{equation}\label{eq:nakamura}
\| (H^D_B(L)+I)^{-2m-1}-(H^N_B(L)+I)^{-2m-1} \|_1 \leq C_m(B) L.
\end{equation}
\end{proposition}


\subsection{Proof of Proposition \ref{dos}}

Throughout the proof, $B$ will be fixed and, for the sake of
simplicity, dropped from the notation. First note that since
\begin{equation*}
	N(\lambda,H^D(L')) \leq N(\lambda,H^D(L)) \leq N(\lambda,H^D(L''))
\end{equation*}
for $L'\leq L\leq L''$ it suffices to prove Proposition \ref{dos} only
for $L\to\infty$ with the flux constraint \eqref{eq:flux}, which we
shall assume henceforth. One has $H^D(L)\geq H^P(L)$ and hence by the
variational principle
\begin{equation*}
	N(\lambda,H^D(L))\leq N(\lambda,H^P(L)).
\end{equation*}
In view of Proposition \ref{torus} this proves the upper bound in \eqref{eq:dos}.

To prove the lower bound we write
\begin{equation*}
	N(\lambda,H^D(L)) = n((\lambda+1)^{-3}, (H^D(L)+I)^{-3})
\end{equation*}
where $n(\kappa,K)$ denotes the number of singular values larger than
$\kappa$ of a compact operator $K$. 
Now by the Ky-Fan inequality \cite[Ch.~11 Sec.~1]{BS} 
for any $\epsilon>0$
\begin{align*}
	& n((\lambda+1)^{-3}, (H^D(L)+I)^{-3}) \\
	& \qquad \geq n((1+\epsilon)(\lambda+1)^{-3}, (H^N(L)+I)^{-3}) \\
	& \qquad \qquad - n(\epsilon(\lambda+1)^{-3}, (H^N(L)+I)^{-3}-(H^D(L)+I)^{-3}).
\end{align*}
We treat the two terms on the RHS separately. The second one can be
estimated using Proposition \ref{nakamura} as follows,
\begin{align*}
	& n(\epsilon(\lambda+1)^{-3}, (H^N(L)+I)^{-3}-(H^D(L)+I)^{-3}) \\
	& \qquad \leq n(\epsilon(\lambda+1)^{-3}, (H^N(L)+I)^{-3}-(H^D(L)+I)^{-3}) \\
	& \qquad \leq \epsilon^{-1}(\lambda+1)^3 \| (H^N(L)+I)^{-3}-(H^D(L)+I)^{-3} \|_1 \\
	& \qquad \leq \epsilon^{-1}(\lambda+1)^3 C_3(B) L.
\end{align*}
On the other hand, writing $\lambda_\epsilon := (1+\epsilon)^{-1/3}(\lambda+1) - 1$ and applying Proposition~\ref{torus} one finds that for $L^2\in 2\pi B^{-1}\N$
\begin{align*}
	n((1+\epsilon)(\lambda+1)^{-3}, (H^N(L)+I)^{-3})
	& = N( \lambda_\epsilon, H^N(L)) \\
	& \geq N( \lambda_\epsilon, H^P(L))
	= L^2 \mathfrak B_0(B, \lambda_\epsilon ).
\end{align*}
Noting that $\lambda_\epsilon<\lambda$ and that $\mathfrak B_0(B,\lambda)$ is left-continuous in $\lambda$ we see that for all sufficiently small $\epsilon>0$ one has
\begin{equation*}
	\mathfrak B_0(B, \lambda_\epsilon) = \mathfrak B_0(B,\lambda).
\end{equation*}
Collecting all the estimates we find that as $L\to\infty$ with $L^2\in 2\pi B^{-1}\N$
\begin{equation*}
	\liminf L^{-2} N(\lambda,H^D(L))
	\geq \mathfrak B_0(B,\lambda).
\end{equation*}
This proves the lower bound in \eqref{eq:dos}.


\section{Proof of the main results}\label{sec:proofs}

\subsection{Non-convex moments for tiling domains}\label{sec:polya}

This subsection is devoted to the proof of Theorem \ref{main2}. We
assume that $\Omega$ is tiling, so we can write
\begin{equation*}
	\R^2 = \bigcup_{n\in\Z^2} \Omega_n
	\qquad\text{up to measure }\, 0
\end{equation*}
where $\Omega_0=\Omega$ and all the $\Omega_n$ are disjoint and congruent to $\Omega$. For $L>0$ let $Q_L=(-L/2,L/2)^2$ and
\begin{equation*}
	J_L:=\{n\in\Z^2 :\ \Omega_n\subset Q_L \},
	\qquad \Omega^L := \operatorname{int}\left(\operatorname{clos}\bigcup_{n\in J_L} \Omega_n \right).
\end{equation*}
We note that
\begin{equation}\label{eq:growth}
	\lim_{L\to\infty} L^{-2} \# J_L = |\Omega|^{-1}.
\end{equation}
Moreover, one has the operator inequalities
\begin{equation*}
	H_B^{Q_L} \leq H_B^{\Omega^L} \leq \sum_{n\in J_L} \oplus H_B^{\Omega_n}.
\end{equation*}
(The first inequality is, of course, understood in terms of the natural embedding $L_2(\Omega^L)\subset L_2(Q_L)$ by extension by zero.) Noting that all the $H_B^{\Omega_n}$ are unitarily equivalent we obtain from the variational principle that
\begin{equation*}
	N(\lambda,H_B^\Omega) \leq (\# J_L)^{-1} N(\lambda,H_B^{Q_L}).
\end{equation*}
The bound \eqref{eq:main2number} follows now from \eqref{eq:growth} and Proposition~\ref{dos} by letting $L$ tend to infinity. This implies also the sharpness of \eqref{eq:main2number}. Indeed, by Proposition~\ref{dos} for any $\epsilon>0$, $B>0$ and $\lambda>0$ there exists a cube $\Omega$ satisfying \eqref{eq:main2counter} for $\gamma=0$.

To prove \eqref{eq:main2moments} we write, in the spirit of \cite{AL},
\begin{equation}\label{eq:liftinggamma}
	\tr(H_B^\Omega-\lambda)_-^\gamma = \gamma \int_0^\infty N(\lambda-\mu,H_B^\Omega) \mu^{\gamma-1}\,d\mu
\end{equation}
and 
\begin{equation}\label{eq:liftinggammab}
	\mathfrak B_\gamma(B,\lambda) = \gamma \int_0^\infty \mathfrak B_0(B,\lambda-\mu) \mu^{\gamma-1} \,d\mu.
\end{equation}
Hence \eqref{eq:main2moments} follows from \eqref{eq:main2number}. Moreover, Proposition \ref{dos}, the formulae \eqref{eq:liftinggamma}, \eqref{eq:liftinggammab} and an easy approximation argument based on \eqref{eq:main2number} imply that
\begin{equation*}
	\lim L^{-2} \tr(H_B^D(L)-\lambda)_-^\gamma = \mathfrak B_\gamma(B,\lambda).
\end{equation*}
As before, this proves the sharpness of the estimate \eqref{eq:main2moments} and concludes the proof of Theorem~\ref{main2}.


\subsection{Convex moments for arbitary domains}\label{sec:elv}

From now on we shall consider arbitrary, not necessarily tiling domains $\Omega$. Our goal is to prove

\begin{theorem}\label{elv2}
	Let $\Omega\subset\R^2$ be a domain of finite measure and let $\gamma\geq 1$. Then for all $B>0$ and $\lambda>0$,
	\begin{equation}\label{eq:elv2proof}
		\tr(H_B^\Omega-\lambda)_-^\gamma \leq \mathfrak B_\gamma(B,\lambda) |\Omega|.
	\end{equation}
\end{theorem}

As we will explain after Corollary \ref{elv} this improves slightly the main result of \cite{ELV}.

\begin{proof}
In the case $\Omega=\R^2$ we write $H_B$ instead of $H_B^\Omega$. By
the variational principle and the Berezin-Lieb inequality (see
\cite{B2}, \cite{Li1} and also \cite{LaS}, \cite{La}), one has for
any non-negative, convex function $\phi$ vanishing at infinity that
\begin{equation*}
	\tr \phi(H_B^\Omega) \leq \tr\chi_\Omega\phi(H_B).
\end{equation*}
Now, if $P_B^{(k)}$ denotes the spectral projection of $H_B$ corresponding to the $k$-th Landau level,
\begin{equation*}
	\phi(H_B) = \sum_{k\in\N_0} \phi(B(2k+1)) P_B^{(k)}.
\end{equation*}
To evaluate the above trace we recall that the integral kernel of
$P_B^{(k)}$ is constant on the diagonal (this follows from the
translation invariance of the Landau Hamiltonian) and has the value
\begin{equation*}
	P_B^{(k)}(x,x)= \frac B{2\pi}.
\end{equation*}
(This is easily seen by diagonalizing $H_B$ with the help of a
harmonic oscillator, see also \cite{F}.) It follows that
$\tr\chi_\Omega P_B^{(k)} = B |\Omega| / 2\pi$. \footnote{To justify
  this, identify the LHS as the square of the Hilbert-Schmidt norm of
  $\chi_\Omega P_B^{(k)}$ and use that $\int |P_B^{(k)}(x,y)|^2\,dy =
  P_B^{(k)}(x,x)$ since $P_B^{(k)}$ is a projection.} This proves that
\begin{equation*}
	\tr \phi(H_B^\Omega) \leq \frac{B |\Omega|} {2\pi} \sum_{k\in\N_0} \phi(B(2k+1)).
\end{equation*}
Specializing to the case $\phi(\mu)=(\mu-\lambda)_-^\gamma$, $\gamma\geq 1$, one obtains the estimate \eqref{eq:elv2proof}. 
\end{proof}


\subsection{Diamagnetic inequalities for the semi-classical symbol}\label{sec:convex}

This subsection illustrates on a semi-classical level the effects that appear when passing from the `magnetic symbol' $\mathfrak B_\gamma(B,\lambda)$ appearing in Theorem \ref{main2} to the `non-magnetic symbol' $L_{\gamma,2}^{\cl} \lambda^{\gamma+1}$ appearing in Theorem \ref{main1}. The convex case $\gamma\geq 1$ appears to be different from the non-convex case $0<\gamma<1$. We shall prove

\begin{proposition}\label{convex}
  Let $\gamma\geq 0$ and $B>0$. Then
  \begin{equation*}
    \sup_{\lambda>0} 
    \frac{\mathfrak B_\gamma(B,\lambda)}{L_{\gamma,2}^{\cl}
    \lambda^{\gamma+1}} 
    = \left\{
    \begin{array}{l@{\qquad\text{if}\;\;}l}
      2 & \gamma=0, \\
      2\left(\frac\gamma{\gamma+1}\right)^\gamma & 0<\gamma<1, \\
      1 & \gamma>1.
    \end{array} 
    \right.
  \end{equation*}
  Moreover, for $0<\gamma<1$ the supremum is attained for $\lambda =
  B(\gamma+1)$ and for $\gamma=0$ the supremum is attained in the
  limit $\lambda\to B+$.
\end{proposition}

We shall need the elementary

\begin{lemma}\label{goingdown}
	Let $\sigma>\gamma\geq 0$ and $\mu>\lambda$. Then for all $E\geq 0$
	\begin{equation*}
		(E-\lambda)_-^\gamma \leq C(\gamma,\sigma) (\mu-\lambda)^{-\sigma+\gamma} (E-\mu)_-^\sigma
	\end{equation*}
	with 	$C(0,\sigma):= 1$ if $\gamma=0$ and $C(\gamma,\sigma):= \sigma^{-\sigma}\gamma^\gamma (\sigma-\gamma)^{\sigma-\gamma}$ if $\sigma>\gamma >0$.
\end{lemma}

For the proof of Lemma \ref{goingdown} one just has to maximize $(\lambda-E)^\gamma (\mu-\lambda)^{\sigma-\gamma}$ as function of $\lambda$ on the interval $(E,\mu)$. 

\begin{proof}[Proof of Proposition \ref{convex}]
	By scaling, we may assume $B=1$. First let $\gamma\geq 1$ and note that the function $\phi(\mu):=(\lambda-\mu)_+^\gamma$ is convex. Then by the mean value property of convex functions
\begin{equation*}
	\phi(2k+1) \leq \frac1{2} \int_{2k}^{2k+2} \phi(\mu)\,d\mu.
\end{equation*}
Summing over $k\in\N_0$ yields the assertion in the case $\gamma\geq 1$.

	Now let $0\leq\gamma<1$. Lemma \ref{goingdown} with $\sigma=1$ together with the inequality that we have already proved implies that for any $\mu>\lambda$
\begin{align*}
	\mathfrak B_\gamma(1,\lambda)			
	& \leq C(\gamma,1) (\mu-\lambda)^{-1+\gamma} \mathfrak B_1(1,\mu) \\
	& \leq C(\gamma,1) L_{1,2}^{\cl} (\mu-\lambda)^{-1+\gamma} \mu^2
\end{align*}
Applying the lemma again, i.e. optimizing in $\mu$, yields the estimate
\begin{align*}
	\tr(H_B^\Omega-\lambda)_-^\gamma \leq R_\gamma L_{\gamma,2}^{\cl} |\Omega| \lambda^2
\end{align*}
where
\begin{align}\label{eq:rgamma}
	R_{\gamma} 
	= \frac{C(\gamma,1)}{C(\gamma+1,2)} \frac{L_{1,2}^{\cl}}{L_{\gamma,2}^{\cl}}
	= 2 \left(\frac\gamma{\gamma+1}\right)^\gamma.
\end{align}
This proves the claimed upper bound on the supremum in the proposition. Choosing $\lambda$ as stated shows that this upper bound is sharp.
\end{proof}

Combining Theorem \ref{elv} with Proposition \ref{convex} we obtain

\begin{corollary}\label{elv}
	Let $\Omega\subset\R^2$ be a domain of finite measure and let $\gamma\geq 1$. Then for all $B>0$ and $\lambda>0$,
	\begin{equation}\label{eq:elvproof}
		\tr(H_B^\Omega-\lambda)_-^\gamma \leq L_{\gamma,2}^{\cl} |\Omega| \lambda^{\gamma+1}.
	\end{equation}
\end{corollary}

Using an idea from \cite{LW2} we now show that \eqref{eq:elvproof}
implies the inequality
\begin{equation}\label{eq:elvevs}
  \sum_{j=1}^N \lambda_j(H_B^\Omega) \geq 2\pi |\Omega|^{-1} N^2
\end{equation}
from \cite{ELV} for the eigenvalues $\lambda_j(H_B^\Omega)$ of
$H_B^\Omega$. For this, we recall the definition of the Legendre
transform of a function $f:\R_+\to\R$,
\begin{equation*}
\tilde f(p) := \sup_{\lambda>0} (p\lambda -f(\lambda),
\end{equation*}
and note that the inequality $f\leq g$ for \emph{convex} functions
$f$, $g$ is equivalent to the reverse inequality $\tilde f\geq\tilde
g$ for their Legendre transforms. Hence an easy calculation shows that
\eqref{eq:elvproof} with $\gamma=1$ is equivalent to the inequality
\begin{equation*}
(p-[p]) \lambda_{[p]+1}(H_B^\Omega) + \sum_{j=1}^{[p]} \lambda_j(H_B^\Omega)
\geq   (4 L_{1,2}^{\cl} |\Omega|)^{-1} p^2,
\qquad p\geq 0,
\end{equation*}
where $[p]$ denotes the integer part of $p$. Choosing $p=N$ one obtains \eqref{eq:elvevs}.

In passing we note that by the same argument the inquality \eqref{eq:elv2proof} (which is stronger than \eqref{eq:elvproof}) is in the case $\gamma=1$ equivalent to the inequality
\begin{equation*}
(p-[p]) \lambda_{[p]+1}(H_B^\Omega) + \sum_{j=1}^{[p]} \lambda_j(H_B^\Omega)
\geq \frac {B^2}{2\pi}
\left( (\tilde p -[\tilde p]) (2[\tilde p]+1) + [\tilde p]^2 \right),
\qquad p\geq 0,
\end{equation*}
where we have set $\tilde p = 2\pi p/(B|\Omega|)$. Estimating the RHS from below by $B^2\tilde p^2 /(2\pi)$ one obtains again \eqref{eq:elvevs}.


\subsection{Non-convex moments for arbitrary domains}

In this subsection we shall prove Theorem \ref{main1}. We deduce the
inequalities \eqref{eq:main1number} and \eqref{eq:main1moments} from
Corollary \ref{elv} in the case $\gamma=1$. The proof is analogous to
that of Proposition \ref{convex}. Indeed, Lemma \ref{goingdown} and
\eqref{eq:elvproof} imply that for any $0\leq\gamma<1$ and for any
$\mu>\lambda$,
\begin{align*}
	\tr(H_B^\Omega-\lambda)_-^\gamma 
	& \leq C(\gamma,1) (\mu-\lambda)^{-1+\gamma} \tr(H_B^\Omega-\mu)_- \\
	& \leq C(\gamma,1) L_{1,2}^{\cl} |\Omega| (\mu-\lambda)^{-1+\gamma} \mu^2.
\end{align*}
Applying the lemma again, i.e. optimizing in $\mu$, yields the estimate
\begin{align*}
	\tr(H_B^\Omega-\lambda)_-^\gamma \leq R_\gamma L_{\gamma,2}^{\cl} |\Omega| \lambda^2
\end{align*}
with $R_\gamma$ as in \eqref{eq:rgamma}. This proves \eqref{eq:main1number} and \eqref{eq:main1moments}.

To prove sharpness of these bounds we note that if $0<\gamma<1$ and $\lambda_\gamma= \gamma+1$ then
\begin{equation*}
	\mathfrak B_\gamma(B,B\lambda_\gamma) = R_\gamma L_{\gamma,2}^{\cl} (B\lambda_\gamma)^{\gamma+1}.
\end{equation*}
Similarly, if $\gamma=0$ one has
\begin{equation*}
	\lim_{\lambda\to 1+}\mathfrak B_0(B,B\lambda) = 2 L_{0,2}^{\cl} B.
\end{equation*}
Hence \eqref{eq:main2counter} implies that for any $\epsilon>0$, $0\leq\gamma<1$ and $B>0$ there exists a cube $\Omega$ satisfying \eqref{eq:main1counter} with $\lambda= B\lambda_\gamma$. This concludes the proof of Theorem~\ref{main1}.


\section{Additional remarks}

\subsection{The three-dimensional case}

The our proof of semi-classical inequalities for the two-dimensional
Dirichlet problem with constant magnetic field is based on two
observations. Firstly, it seems to be appropriate to estimate eigenvalue
sums $\tr(H_B^\Omega-\lambda)_-^\gamma$ in terms of the respective
average of the {\em
  magnetic} symbol $\mathfrak B_\gamma(B,\lambda)$. Indeed, the bound
\[\tr(H_B^\Omega-\lambda)_-^\gamma\leq \mathfrak
  B_\gamma(B,\lambda)|\Omega|\,,
\]
which holds true for arbitrary $\Omega$ for $\gamma\geq 1$ and for
tiling domains for $\gamma\geq 0$,
is sharp, 
since the ratio
\[
\frac{\tr(H_B^\Omega-\lambda)_-^\gamma}{\mathfrak B_\gamma(B,\lambda)|\Omega|}
\]
can be made arbitrary close to $1$ by a suitable 
choice of (large) $\Omega$.

Secondly, the average of the {\em magnetic} symbol 
satisfies a sharp estimate
by the standard non-magnetic phase space average from above
\[
\mathfrak B_\gamma(B,\lambda)\leq
L_{\gamma,2}^{\cl}\lambda^{\gamma+1}\,.
\]
for $\gamma\geq 1$ only. For $\gamma<1$ this leads in conjunction with
the asymptotic argument to the counterexamples stated above.

As we shall see in this subsection, in the \emph{three-dimensional
  case} the asymptotic behavior of eigenvalue moments is still
  governed by the average of a suitable magnetic symbol. However, this
  average will not exceed the corresponding classical phase space
  average for all $\gamma\geq 1/2$. Therefore our approach produces
  counterexamples to inequalities with semi-classical constants only
  for $0\leq\gamma<1/2$. We shall discuss this below in more detail.

Let $\Omega\subset\R^3$ be a domain of finite measure and consider for $B>0$ the self-adjoint operator 
\begin{equation*}
	H_B^\Omega := (D-BA)^2
	\qquad\text{in } L_2(\Omega)
\end{equation*}
with Dirichlet boundary conditions where now
\begin{equation*}
	A(x) := \frac12 (-x_2,x_1,0)^T.
\end{equation*}
In the three-dimensional case the magnetic symbol is define as
\begin{align*}
\mathfrak B_\gamma^{(3)}(B,\lambda) 
& := (2\pi)^{-1} \int_\R \mathfrak B_\gamma(B,\lambda-|\xi|^2)\,d\xi \\
& = \frac{\Gamma(\gamma+1)}{\Gamma(\gamma+3/2)} \frac{B}{4\pi^{3/2}} 
\sum_{k\in\N_0} \left(\lambda-B(2k+1)\right)_+^{\gamma+1/2}.
\end{align*}
Similarly as in Subsection \ref{sec:elv} one proves that
\begin{equation}\label{eq:elv3d}
	\tr(H_B^\Omega-\lambda)_-^\gamma \leq \mathfrak B_\gamma^{(3)}(B,\lambda) |\Omega|,
	\qquad \gamma\geq 1.
\end{equation}
Put
\begin{equation*}
L_{\gamma,3}^\cl := (2\pi)^{-3} \int_{\{|\xi|<1\}} (1-|\xi|^2)^{\gamma}\,d\xi
= \frac1{8\pi^{3/2}} \frac{\Gamma(\gamma+1)}{\Gamma(\gamma+5/2)}
\end{equation*}
By the same argument as in Proposition \ref{convex} one has
\begin{equation}\label{eq:convex3d}
B_\gamma^{(3)}(B,\lambda) \leq L_{\gamma,3}^\cl \lambda^{\gamma+3/2},
\qquad \gamma\geq 1/2,
\end{equation}
and hence
\begin{equation*}
	\tr(H_B^\Omega-\lambda)_-^\gamma \leq L_{\gamma,3}^\cl \lambda^{\gamma+3/2} |\Omega|,
	\qquad \gamma\geq 1.
\end{equation*}

Again the quantity $\mathfrak B_0^{(3)}(B,\lambda)$ arises as the density of states. More precisely, if $Q_L :=(-L/2,L/2)^3$ then a three-dimensional version of Proposition~\ref{nakamura} allows to prove that
\begin{equation}\label{eq:dos3d}
	\lim_{L\to\infty} L^{-3} N(\lambda, H_B^{Q_L}) = \mathfrak B_0^{(3)}(B,\lambda).
\end{equation}
This implies as in the two-dimensional case

\begin{theorem}\label{main23d}
	Let $\Omega\subset\R^3$ be a tiling domain of finite measure. Then for all $B>0$ and $\lambda>0$
	\begin{equation}\label{eq:main23dnumber}
		N(\lambda, H_B^\Omega) \leq \mathfrak B_0^{(3)}(B,\lambda) |\Omega|
	\end{equation}
	and
	\begin{equation}\label{eq:main23dmoments}
		\tr(H_B^\Omega-\lambda)_-^\gamma \leq \mathfrak B_\gamma^{(3)}(B,\lambda) |\Omega|,
		\qquad 0<\gamma<1,
	\end{equation}
	and these estimates cannot be improved. More precisely, for any $0\leq\gamma<1$, $\epsilon>0$, $B>0$, $\lambda>0$ there exists a cube $\Omega$ such that
	\begin{equation}\label{eq:main23dcounter}
		\tr(H_B^\Omega-\lambda)_-^\gamma \geq (1-\epsilon) \mathfrak B_\gamma^{(3)}(B,\lambda) |\Omega|.
	\end{equation}
\end{theorem}

The estimates \eqref{eq:main23dnumber}, \eqref{eq:main23dmoments} and Proposition \ref{convex} imply that for tiling domains $\Omega$ and for $0\leq\gamma<1/2$,
\begin{equation}\label{eq:main13dmoments}
		\tr(H_B^\Omega-\lambda)_-^\gamma \leq R_{\gamma+1/2} L_{\gamma,3}^\cl \lambda^{\gamma+3/2}
\end{equation}
with $R_\gamma$ as in Theorem \ref{main1}. Moreover, the asymptotics \eqref{eq:dos3d} imply that this constant can not be replaced by a smaller one. However, in contrast to the two-dimensional case we do not know whether the constant in this estimate has to be further increased if non-tiling domains are considered.

On the other hand, \eqref{eq:main23dmoments} and \eqref{eq:convex3d} imply that for tiling domains $\Omega$ and for $\gamma\geq 1/2$,
\begin{equation}\label{eq:main13dhighmoments}
		\tr(H_B^\Omega-\lambda)_-^\gamma \leq L_{\gamma,3}^\cl \lambda^{\gamma+3/2}.
\end{equation}
We do not know whether the constant in this estimate has to be increased if $1/2\leq\gamma<1$ and if non-tiling domains are considered.

The method of Appendix \ref{app:blymagnonsharp} allows to deduce from \eqref{eq:elv3d} (probably non-sharp) estimates on $\tr(H_B^\Omega-\lambda)_-^\gamma$ for $0\leq\gamma<1$ and arbitrary $\Omega$. We omit the details.

Another remark concerns domains with product structure.

\begin{proposition}\label{product}
	Let $\omega\subset\R^2$ be a domain of finite measure, $I\subset\R$ a bounded open interval and $\Omega:=\omega\times I$, and let $\gamma\geq1/2$. Then for all $B>0$ and $\lambda>0$,
\begin{equation}\label{eq:product}
	\tr(H_B^\Omega-\lambda)_-^\gamma \leq \mathfrak B_\gamma^{(3)}(B,\lambda) |\Omega|.
\end{equation}
\end{proposition}

It follows from \eqref{eq:convex3d} that for domains of this form and for $\gamma\geq 1/2$ one has also \eqref{eq:main13dhighmoments}. 

\begin{proof}
We follow Laptev's lifting idea \cite{La}. By separation of variables
we can write
\begin{equation*}
	\tr(H_B^\Omega-\lambda)_-^\gamma 
	= \sum_{n\in\N} \tr\left(H_B^\omega + \left(\frac{\pi n}{|I|}\right)^2-\Lambda\right)_-^\gamma.
\end{equation*}
P\'olya's estimate on an interval states that
\begin{equation*}
	\sum_{n\in\N} \left(\left(\frac{\pi n}{|I|}\right)^2-E \right)_-^\gamma
	\leq L_{\gamma,1}^\cl |I| E^{\gamma+1/2}
\end{equation*}
where
\begin{equation*}
	L_{\gamma,1}^\cl := \frac1{2\sqrt\pi} \frac{\Gamma(\gamma+1)}{\Gamma(\gamma+3/2)}.
\end{equation*}
Hence
\begin{equation*}
	\tr(H_B^\Omega-\lambda)_-^\gamma 
	\leq L_{\gamma,1}^\cl |I| \tr(H_B^\omega-\lambda)_-^{\gamma+1/2}.
\end{equation*}
Applying Theorem \ref{eq:elv2proof} and noting that
\begin{equation*}
L_{\gamma,1}^\cl \mathfrak B_{\gamma+1/2}(B,\lambda) = \mathfrak B_\gamma^{(3)}(B,\lambda)
\end{equation*}
completes the proof.
\end{proof}


\subsection{The role of the integrated density of states}

Our reasoning in Subsection \ref{sec:polya} has shown that the
important idea in P\'olya's proof is not the high energy limit, but
the large domain limit. (In the non-magnetic case these two limits are
equivalent by scaling.) The large domain limit corresponds to the
passage to the density of states.

More generally, one can prove the following. 
For the sake of simplicity we return to the two-dimensional case.
Assume that
$\Omega\subset\R^2$ is a tiling domain and write
\begin{equation*}
	\R^2 = \bigcup_{n\in\Z^2} \Omega_n
	\qquad\text{up to measure }\, 0.
\end{equation*}
Here $\Omega_0=\Omega$ and all the $\Omega_n$ are disjoint with
$\Omega_n = G_n\Omega$ for $G_n$ a composition of a translation and a
rotation. Let $V$ and $A$ be a sufficiently regular real-valued
function, respectively vectorfield on $\Omega$ and consider the
self-adjoint operator $H^\Omega := (D-A)^2+V$ with Dirichlet boundary
conditions in $L_2(\Omega)$. 

We extend $V$ and $A$ to the whole plane in such a way that $V(x)=V(G_n^{-1}x)$ and $\curl A(x) = \curl A(G_n^{-1}x)$ for $x\in\Omega_n$. This allows to define a self-adjoint operator $H :=(D-A)^2+V$ in $L_2(\R^2)$. Our main assumption is that this operator possesses an integrated density of states at a certain $\lambda\in\R$, i.e., there
exists a number $n(\lambda)\geq 0$ such that
\begin{equation}\label{eq:exids}
\lim_{L\to\infty} L^{-2} N(\lambda, H^{Q_L}) = n(\lambda).
\end{equation}
Here as before, $Q_L=(-L/2,L/2)$. Under this assumption one has for this given value of $\lambda$ the
P\'olya estimate
\begin{equation*}
N(\lambda, H^\Omega) \leq n(\lambda) |\Omega|.
\end{equation*}
This is proved in the same way as Theorem \ref{main2}.

A special case is when the $G_n$ are translations. If the flux of $\curl A$ through $\Omega$ vanishes, then $A$ can be chosen periodic and one can apply Floquet theory. In this case it is well-known that the limit \eqref{eq:exids} exists for any $\lambda$ and defines a non-negative, increasing and left-continuous function $n$ on $\R$. A more general case is that of $G_n$'s which correspond to almost-periodic tilings. The existence of the limit \eqref{eq:exids} in the almost-periodic case under broad conditions on the coefficients has been proved, e.g., in \cite{S}.


\begin{appendix}


\section{The case of an arbitrary magnetic field}\label{app:blymagnonsharp}

In this section we consider an \emph{arbitrary} magnetic field $A\in
L_{2,\mbox{loc}}(\overline{\Omega})$ with $\Omega\subset\R^d$ in any
dimension $d\geq 2$ and define $H_\Omega(A)=(D-A)^2$ on $\Omega$ with
Dirichlet boundary conditions. We shall prove the estimate
\begin{equation}\label{eq:blymagnonsharpd}
  \tr(H_\Omega(A)-\lambda)_-^\gamma 
  \leq \rho_{\gamma,d} L_{\gamma,d}^{\cl} \lambda^{\gamma+d/2} |\Omega|,
  \qquad 0\leq \gamma<3/2.
\end{equation}
Here
\begin{equation*}
  \rho_{\gamma,d} := 	
  \frac{\Gamma(5/2)\, \Gamma(\gamma+d/2+1)}{\Gamma((5+d)/2)\, \Gamma(\gamma+1)}
  3^{-3/2} (3+d)^{(3+d)/2} (2\gamma)^\gamma (2\gamma+d)^{-\gamma-d/2}
\end{equation*}
and
\[
L_{\gamma,d}^{\cl}
=\frac{\Gamma(\gamma+1)}{2^d\pi^{d/2}\Gamma(\gamma+\frac{d}{2}+1)}\,.
\]
Note that for $d=2$ the constant $\rho_{\gamma,d}$ equals
\[
\rho_{\gamma,2} := (5/3)^{3/2} (\gamma/(\gamma+1))^\gamma,
\]
and it follows from our main result that this is off at most by a
factor $(5/3)^{3/2}/2\approx 1.0758\,.$

To prove \eqref{eq:blymagnonsharpd} we recall the sharp Lieb-Thirring
bound on the negative spectrum of a magnetic Schr\"odinger operator
$H_{\R^d}(A,V)=(D-A)^2-V$ in $\R^d$ from \cite{LW},
\[
\tr (H_{\R^d}(A,V))_-^{3/2}
\leq L_{3/2,d}^{\cl} \int_{\R^d} V(x)^{(3+d)/2}_+\,dx\,.\]
Here we extend the given magnetic vector potential $A$ on
$\overline\Omega$ by $0$ to $\R^d$. Since the negative eigenvalues of
$H_\Omega(A)-\mu$ are not below those of $H_{\R^d}(A,V)$ with
$V(x):=\mu$ for $x\in\Omega$ and $V(x):=0$ for
$x\in\R\setminus\Omega$, we find
\[\tr (H_{\Omega}(A)-\mu)_-^{3/2}
\leq \tr (H_{\R^d}(A,V))_-^{3/2}
\leq L_{3/2,d}^{\cl}|\Omega|\mu^{(3+d)/2}\,.
\]
Lemma \ref{goingdown} with $\sigma=3/2$ shows now that for $0\leq\gamma<3/2$
\begin{align*}
	\tr(H_{\Omega}(A)-\lambda)_-^\gamma 
	& \leq C(\gamma,3/2) (\mu-\lambda)^{-3/2+\gamma} \tr(H_{\Omega}(A)-\mu)_-^{3/2} \\
	& \leq C(\gamma,3/2) L_{3/2,d}^{\cl} |\Omega| (\mu-\lambda)^{-3/2+\gamma} \mu^{(3+d)/2}
\end{align*}
for any $\mu>\lambda$. Again by this lemma, i.e., optimizing in $\mu$,
we get \eqref{eq:blymagnonsharpd} with excess factor 
\begin{align*}
	\rho_{\gamma,d} 
	& = \frac{L_{3/2,d}^{\cl}}{L_{\gamma,d}^{\cl}} 
	\frac{C(\gamma,3/2)}{C(3/2-\gamma,(3+d)/2)}
	= \frac{L_{3/2,d}^{\cl}}{L_{\gamma,d}^{\cl}}
	\frac{(3+d)^{(3+d)/2}}{3^{3/2}} \frac{(2\gamma)^\gamma}{(2\gamma+d)^{\gamma+d/2}}.
\end{align*}
Recalling the definition of $L_{\gamma,d}^{\cl}$ we obtain the claimed statement.

Besides the case of a homogeneous magnetic field also the case of a
$\delta$-like magnetic field (Aharonov-Bohm field) has received
particular attention. In \cite{FH} the above value of the excess
factor $\rho_{\gamma,2}$ could be slightly improved for this case, but
it is still unknown whether or not this factor can be chosen
one for $0\leq\gamma<3/2$.

\end{appendix}

\subsection*{Acknowledgments} This work had its gestation at the
workshop `Low eigenvalues of Laplace and Schr\"odinger operators'
which was held at AIM in May 2006. The support of AIM is gratefully
acknowledged. This work has been partially supported by DAAD grant
D/06/49117 (R. F.), NSF grant DMS 0600037 (M. L.) and DFG grant
WE-1964/2-1 (T.~W.), as well as by the DAAD-STINT PPP program (R. F. and T. W.).


\bibliographystyle{amsalpha}

\begin{thebibliography}{ELV}

\bibitem[AL]{AL} M. Aizenman, E. Lieb, \textit{On semiclassical
  bounds for eigenvalues of Schrödinger operators}. Phys. Lett. A
  \textbf{66} (1978), no. 6, 427--429.
\bibitem[B1]{B} F.A.~Berezin, \textit{Covariant and contravariant
  symbols of operators} [Russian]. Math. USSR Izv. {\bf 6} (1972),
  1117--1151.
\bibitem[B2]{B2} F.~A.~Berezin, \textit{Convex functions of operators}
  [Russian]. Mat. Sb. \textbf{88} (1972), 268--276. 
\bibitem[BS]{BS} M.S. ~Birman, M.Z. ~Solomjak, \textit{Spectral Theory
    of Self-Adjoint Operators in Hilbert Space}. D. Reidel Publishing
  Company, Dortrecht, Holland (1987) 
\bibitem[CV]{CV} Y.~Colin de Verdiere, \textit{L'asymptotique de Weyl pour les bouteilles magnetiques} [French]. Comm. Math. Phys. \textbf{105} (1986), 327--335.
\bibitem[ELV]{ELV} L.~Erd\"os, M.~Loss, V.~Vugalter,
  \textit{Diamagnetic behavior of sums of Dirichlet eigenvalues}. {\em
    Ann. Inst. Fourier} {\bf 50} (2000), 891--907.
\bibitem[F]{F} V.~Fock, \textit{Bemerkung zur Quantelung des
  harmonischen Oszillators im Magnetfeld} [German]. Z. Physik
  \textbf{47} (1928), 446--448.
\bibitem[FH]{FH} R.~L.~Frank, A.~M.~Hansson, \textit{Eigenvalue
    estimates for the Aharonov-Bohm operator in a domain}. Submitted.
\bibitem[L]{La} A. Laptev, \textit{Dirichlet and Neumann eigenvalue
    problems on domains in Euclidean
    spaces}. J. Funct. Anal. \textbf{151} (1997), no. 2, 531--545.
\bibitem[LS]{LaS} A.~Laptev, Yu.~Safarov, \textit{A generalization of
    the Berezin-Lieb inequality}. Amer. Math. Soc. Transl. (2)
  \textbf{175} (1996), 69--79.
\bibitem[LW1]{LW} 
  A. Laptev, T. Weidl, \textit{Sharp Lieb-Thirring inequalities in
  high dimensions}. Acta Math. \textbf{184} (2000), no. 1, 87--111.
\bibitem[LW2]{LW2} 
  A. Laptev, T. Weidl, \textit{Recent results on Lieb-Thirring
  inequalities}. Journ\'ees "\'Equations aux D\'eriv\'ees Partielles"
  (La Chapelle sur Erdre, 2000), Exp. No. XX, Univ. Nantes, Nantes,
  2000.
\bibitem[LY]{LY} P.~Li, S-T.~Yau, \textit{On the Schr\"odinger equation and
    the eigenvalue problem}. Comm. Math. Phys. {\bf 88} (1983),
    309--318.
\bibitem[L1]{Li1} E.~H.~Lieb, \textit{The classical limit of quantum spin
    systems}. Comm. Math. Phys. \textbf{31} (1973), 327--340.
\bibitem[L2]{Li2} E.~H.~Lieb, \textit{The number of bound states of
  one-body Schr\"rodinger operators and the Weyl
  problem}. Proc. Sym. Pure Math. \textbf{36} (1980), 241--252.
\bibitem[LT]{LT} E.~H.~Lieb, W. Thirring, \textit{Inequalities for the
  moments of the eigenvalues of the Schr\"odinger Hamiltonian and
  their relation to Sobolev inequalities}. Studies in Mathematical
  Physics, 269--303. Princeton University Press, Princeton, NJ, 1976.
\bibitem[M]{M}
  G.~Metivier, \textit{Valeurs propres de probl\`emes aux limites
  elliptiques irr\'eguliers}. Bull. Soc. Math. France,
  Mem. \textbf{51-52} (1977), 125--229.
\bibitem[N]{N} S. Nakamura, \textit{A remark on the Dirichlet-Neumann decoupling and the integrated density of states}. J.~Funct. Anal. \textbf{179} (2001), no. 1, 136--152.
\bibitem[P]{P} G.~P\'olya, \textit{On the eigenvalues of vibrating
  membranes}. Proc. London Math. Soc. {\bf 11} (1961), 419--433.
\bibitem[RS]{RS} M.~Reed, B.~Simon, \textit{Methods of Modern
  Mathematical Physics}, vol. 4. Academic Press, 1978.
\bibitem[R1]{R1} 
  G.~V.~Rozenblyum \textit{On the eigenvalues of the first
    boundary value problem in unbounded
  domains}. Math. USSR-Sb. \textbf{18} (1972), 235--248.
\bibitem[R2]{R2} 
  G.~V.~Rozenblyum, \textit{Domination of semigroups and estimates for
  eigenvalues}. St. Petersburg Math. J. \textbf{12} (2001), no. 5,
  831--845.
\bibitem[S]{S} M.~A.~Shubin, \textit{Spectral theory and the index of
    elliptic operators with almost-periodic coefficients}. Russian
  Math. Surveys. \textbf{34} (1979), no. 2, 109--158.
\bibitem[W]{W} H.~Weyl, \textit{Das asymptotische Verteilungsgesetzt
  der Eigenwerte linearer partieller
  Differentialgleichungen}. Math. Ann. \textbf{71} (1911), 441--469.

\end{thebibliography}

\end{document}